\documentclass[12pt]{article}
\usepackage[english]{babel}
\usepackage[utf8]{inputenc}
\usepackage{fancyhdr}
\usepackage{amsmath, amsfonts, amssymb, amsthm, geometry, hyperref, graphicx, mathtools, subcaption, array, enumitem, ulem}
\usepackage{tikz, verbatim}
\tikzstyle{every node}=[circle, draw, fill=black!50,
                        inner sep=0pt, minimum width=4pt]
\usepackage[skip=10pt,font=small]{caption}
\usepackage{tabularx}
\usepackage[symbol]{footmisc}
\usepackage[wide]{sidecap}
\usepackage{fullpage}
\usepackage{caption}
\usepackage{parskip}
\graphicspath{}
\numberwithin{equation}{section}
\newlist{Properties}{enumerate}{3}
\newtheorem{claim}{Claim}
\theoremstyle{definition}
\newtheorem{definition}{Definition}[section]

\usepackage{color}

\newcommand{\mS}{\mathcal S}
\newcommand{\cZ}{\mathcal Z}

\newcommand{\cP}{\mathcal P}
\newcommand{\cD}{\mathcal D}

\newcommand{\Nm}{N_{\mathrm{max}}}

\newcommand{\link}{\prec\!\!\ast\;}
\newcommand{\ql}{QL}

\newcommand{\pql}{PQL}

\newcommand{\cDnk}{\mathcal{D}_n^k}
\newcommand{\cQnk}{\mathcal{Q}_n^k}

\newcommand{\cPnk}{\mathcal{P}_{n}^k}

\newcommand{\cQpnk}{\mathcal{Q}_{p,n}^k}
\newcommand{\Ppql}{\mathcal{V}_n}

\newcommand{\vq}{\vec{q}}

\newcommand{\lpp}{\mathcal{Z}_n^{(L)}}
\newcommand{\cLn}{\mathcal{Q}^2_n}
\newcommand{\cLpn}{\mathcal{Q}^2_{p,n}}

\newcommand{\vqs}{\mathcal{V}^*_{\vq}}
\newcommand{\krd}{\mathcal{T}}
\newcommand{\sts}{path-sum}
\newcommand{\cc}{c}
\newcommand{\cs}{\hat{c}}
\begin{document}

\date{}
% \title{Probing Non-Manifold-Like Causal Sets}
\title{Entropy and the Link Action in the\\ Causal Set Path-Sum}
\author{Abhishek Mathur, Anup Anand Singh and Sumati Surya\\ {\small\it Raman Research Institute, CV Raman Ave, Sadashivanagar, Bangalore 560080}}
\maketitle
\abstract{In causal set theory the gravitational path integral is replaced by a {\sts} over a sample space $\Omega_n$ of $n$-element causal sets. The contribution from non-manifold-like orders dominates $\Omega_n$ for large $n$ and therefore must be tamed by a suitable action in the low energy limit of the theory. We extend the work of Loomis and Carlip on the contribution of sub-dominant bilayer orders to the causal set {\sts} and show that the ``link action" suppresses the dominant Kleitman-Rothschild orders for the same range of parameters.}
\setlength{\parskip}{0.3em}
\section{Introduction}
In any theory of quantum gravity the transition from the deep quantum regime to the semi-classical regime requires the
suppression of non-classical ``quantum spacetimes". In the causal set approach to quantum gravity {\cite{Bombelli:1987aa,Surya:2019ndm}}, a quantum spacetime corresponds to a {\sl{causal set}} or locally finite order\footnote{In this work we shall use the term ``order" instead of ``poset".}, and the path integral is replaced by a {\sts} over a sample space of causal sets.

As shown in \cite{kleitman1}, if one were to randomly pick an order from the sample space $\Omega_n$ of finite $n$-element orders, it would overwhelmingly be a ``Kleitman-Rothschild'' (KR) order as $n$ becomes very large. A KR order has three levels with approximately $n/4$ elements in the top and bottom levels, $n/2$ elements in the middle level, and such that every element in the top level and the bottom level is linked to approximately half of the elements in the middle level.  In causal set theory (CST) a causal set is said to be manifold-like only if it can be obtained from a (typical) Poisson sprinkling into a spacetime. Thus KR orders are far from manifold-like. This poses a challenge to CST, since continuum-like dynamics must arise from the fundamentally discrete dynamics in the semi-classical limit. The number of KR orders goes as $\sim 2^{\frac{n^2}{4}+\frac{3n}{2} + o(n)}$, and is the dominant entropic contribution to the CST {\sts}.  This entropy therefore needs to be suppressed in the semi-classical limit by an appropriate choice of action.

In addition to the KR orders, as shown by Dhar \cite{dhar1}, there is a hierarchy of sub-dominant ``$k$-layer"  orders
which are also not manifold-like. Of these, the next dominant contribution to the entropy of the {\sts} comes from the
{bilayer} orders. In \cite{carlip}, {Loomis and Carlip showed that the discrete Einstein-Hilbert or
  Benincasa-Dowker-Glaser (BDG) action} \cite{benincasa,glaser1,glaser2}  suppresses all {bilayer} orders in the
CST {\sts}.  For  {{bilayer}} orders the BDG action simplifies to the ``link''  action {$\mS_L\propto N_0 $,
  where $N_0$ counts the number of links in the causal set, thus making the counting of ``iso-action'' causal sets
  considerably simpler. This is of course not true in general. }  In this work we explicitly show how the analysis in
\cite{carlip} carries over to a larger class of causal
sets which include the KR orders, when the BDG action is replaced
by the link action. While the {BDG} action limits to the {Einstein-Hilbert} action as $n\rightarrow\infty$, the link action can be considered a more natural choice from the order theoretic perspective. These two points of view
are in principle compatible, since one expects higher order corrections to the Einstein-Hilbert action. The (linear) link action can be considered to be one such possible correction\footnote{See \cite{astrid} for a discussion on using functional renormalisation group ideas to causal set actions.}.

In Section \ref{sec:prelim}, we first define what is meant by a {\sl level} as well as a {\sl $k$-layer}. In the
literature these terms are often used interchangeably, although they are distinct. More confusingly, the same
terminology is used in subtly different ways, for example in \cite{dhar1} and \cite{promel} versus \cite{kleitman1}. We
distinguish these different usages by defining two new types of layers: the quasi-layers ({\ql}s) and the
pseudo-quasi-layers ({\pql}s) and the associated classes of causal sets. In Section \ref{sec:csaction} we give a brief review of both the {BDG} and the link actions.

We present our main results in Section \ref{sec:3p} after reviewing the work of \cite{carlip} in the language of {\ql}s and
{\pql}s in Section \ref{sec:revlc}. In Section \ref{sec:3pql} we show that a subset $\cP^*_{\vq,p,n}\subset\Omega_n$
which contains ``typical" naturally labelled KR orders with fixed level sizes is suppressed by the link action to
leading order. This is also true of any naturally labelled $k$-{\pql} order with fixed layer size. Surprisingly, while
the counting arguments are different from those used in \cite{carlip}, we find the leading order contribution to the {\sts} to
be the same, so that  their results carry over trivially. In Section \ref{section:klevels}, we expand the analysis to $k$-{\ql} orders for any $k$ (and therefore all KR orders when $k>2$) and again find that the leading order contribution is the same as that found in \cite{carlip}.

In order to emphasise the non-triviality of these results, we use a very different subset of $\Omega_n$ in Section \ref{sec:krlike}. This ``KR + dust" subset consists of ``typical" $n'$-KR orders with fixed layer sizes for $n'\leq n$, with the remaining $n-n'$ elements forming an antichain or {\sl{dust}}. Unlike our earlier results, here we find that the leading order contribution to the CST {\sts} is {\it{not}} suppressed by the link action. We summarise our results in Section \ref{sec:con}, and discuss some of the open questions.

\section{Preliminaries}\label{sec:prelim}
We consider the sample space $\Omega_n$ of finite $n$-element posets or {\sl{orders}}, which are {\sl{labelled}} over the set of $n$ integers as in \cite{kleitman1,dhar1,promel,brightwell}\footnote{As discussed in \cite{henson}, the labelling introduces a factor of at most $n!$ which is sub-dominant to the entropic factor of $2^{n^2/4}$. Hence much of our analysis carries over to the unlabelled case, unless stated otherwise, as in Sec.~\ref{sec:3pql}.}. An order $\cc\in\Omega_n$ is said to be {\sl{naturally labelled}} if $\forall\;e_r,e_s\in \cc,\;e_r\prec e_s\Rightarrow r<s$.

As in standard CST terminology: (i) a {\sl minimal element} in $\cc$ is one with no preceding element (ii) a {\sl link}
$\link$ is a relation not implied by transitivity, i.e.,  $e_r \link e_s$ if  $\nexists\,e_t\in \cc $ such that $ e_r
\prec e_t \prec e_s$ (iii) an {\sl{order interval}} is the set $I[e_r,e_s]:=\{e_t\in \cc\,|\,e_r\prec e_t\prec e_s\}$
and is called a {\sl{$j$-element order interval}} if $|I[e_r,e_s]|=j$.

{In what follows, we give the standard definitions of a {\sl level} in a causal set as well as some new definitions of
  different types of {\sl layers} in a causal set. We  find this new  classification of  causal sets useful in
  obtaining the cardinality of iso-action causal sets. We  give several examples of these along the way. All the
  figures shown are {\sl Hasse diagrams} which only show the link relations (the rest being implied by
  transitive closure), with the directed relation going from the bottom of the figure (past) to the top (future).}

\begin{definition} \label{def:level}
The {\sl{level}} $L_j$ of an order $\cc$, with  $j = 1, 2, 3, \dots k$, is the set of minimal elements that remain after
deleting all elements in levels $L_m$,  $m < j$. In particular, $L_1$ contains all the minimal elements of $\cc$ \cite{kleitman1} (see Fig.~\ref{fig:levels}).
\end{definition}
\begin{figure}[h]
\begin{center}
\begin{tabular}{c}
\includegraphics{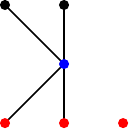}
\end{tabular}
\caption{The Hasse diagram of a  $6$-element order with $3$ levels.  Level $1$ elements are coloured red, level $2$ blue
  and level $3$ black. }\label{fig:levels}
\end{center}
\end{figure}
{As is obvious from this definition, a {level} can be assigned to {\it any} causal set.}

\begin{definition}\label{def:kr}
A {\sl{KR order}} has three levels $L_1,L_2,L_3$ satisfying the following properties (see Fig.~\ref{fig:krorders}):
\begin{enumerate}%[label=Property \arabic*.,itemindent=*]
\item $|L_1|,|L_3|=n/4+o(n)$ and $|L_2|=n/2+o(n)$.
\item $e_r\link e_s$ and $e_r\in L_j$ implies $e_s\in L_{j+1}$.
\item Each element in a level $L_j$ is connected to asymptotically half of the elements in $L_{j-1}$ and half of the
  elements in $L_{j+1}$.
\item For all $e_r\in L_1$ and $e_s\in L_3$, $e_r\prec e_s$.%\red{{\bf Look for this reference.}}
%\red{[[Dhar, Brigthwell, Promel and others include this property but Kleitman and Rothschild do not -- trying to make sure this is the case]]}
\end{enumerate}
\end{definition}
{\begin{figure}[h]
\begin{center}
\begin{tabular}{c}
\includegraphics{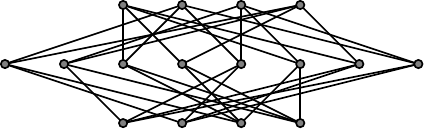}
\end{tabular}
\caption{An example of a  $16$-element order, which satisfies the KR property.}\label{fig:krorders}
\end{center}
\end{figure}}
Def.~\ref{def:kr}-4 is not explicitly stated in the original paper of Kleitman and Rothschild \cite{kleitman1}. However
since the dominant contribution comes from those orders satisfying all the conditions in Def.~\ref{def:kr}, this
additional condition was imposed in \cite{dhar1,promel,brightwell}.  Dhar showed that in addition to the KR orders,
there is a hierarchy of {\sl k-layer orders} which are subdominant \cite{dhar1}.

\begin{definition}\label{def:layer-map}  We will refer to the  map $ \zeta: \{e_1, e_2, \ldots e_n\} \rightarrow \{ 1,
\ldots, k \} $ as a  {\sl $k$-layer-map}.
\end{definition}

\begin{definition}\label{def:layer}
\cite{dhar1,promel,brightwell} In a  {\sl{k-layer order}}  $\cc\in\Omega_n$ it is possible to assign a  {\sl $k$-layer-map} $\zeta(e_r) \in
\{1,2,\dots,k\}$  to each element $e_r\in \cc$, such that:
\begin{enumerate}%[label=Property \arabic*.,itemindent=*]
\item $e_r \prec e_s \Rightarrow \zeta(e_r)<\zeta(e_s)$.
\item $\zeta(e_s)>\zeta(e_r)+1 \Rightarrow e_r \prec e_s$.
\end{enumerate}
Let $\cD_n^k$ denote the set of these orders.
\end{definition}

 It is clear that the class of $k$-layer orders is special and that not every causal set is a $k$-layer
 order for any choice of $k$, as shown in Fig.~\ref{fig:layerorders}.
{\begin{figure}[h]
\begin{center}
\begin{tabular}{c c c}
\includegraphics{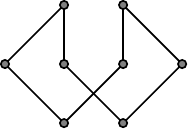}
\includegraphics{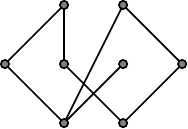}
\includegraphics{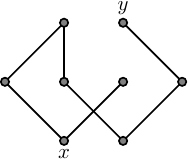}
\end{tabular}
\caption{The first two orders are  $3$-layer orders while the third one is not a $k$-layer order for {\it any}
  $k$. While it is possible to assign a map $\zeta: c \rightarrow \mathbb N$ which satisfies Condition 1 of Def. \ref{def:layer}, as shown in
  the figure, it cannot
  satisfy Condition 2,  since element $x$ with $\zeta(x)=1$ and element $y$ with  $\zeta(y)=3$ are not related. \label{fig:layerorders}}
\end{center}
\end{figure}}
%In this work we consider another distinct class $\cQnk$ of orders which also includes the $KR$ orders.% For this we
%first need a few more definitions.

In this work we will need to consider two new distinct classes of causal sets, also constructed via $k$-layer-maps, which
are distinct from the set $\cD_n^k$.

\begin{definition}\label{def:cd}
A subset $\cs\subset \cc$ is {\sl{causally disconnected}} if there exists no relation between elements of $\cs$ and its complement $\cs^c$ in $\cc$. $\cs$ is an {\sl{irreducible causally disconnected}} subset of $\cc$ if further, it contains no non-trivial causally disconnected proper subsets.
\end{definition}
% \begin{definition}
% For a causal set $C$, $x,y\in C$ are a connected via a {\sl{link}} directed from $x$ to $y$ ($x\link y$) if $x\prec y$ and $|\{z \in C: x \prec z \prec y \}| = 0$.
% \end{definition}

\begin{definition}\label{def:ql}
In a {\sl{$k$-quasi-layer ($k$-QL)}} order $\cc\in\Omega_n$  it is possible to assign a $k$-layer-map $\eta$, such that:
\begin{enumerate}%[label=Property \arabic*.,itemindent=*]
\item For $e_r, e_s\in \cc$, if $e_r \link e_s$ then ${\eta}(e_s) = {\eta}(e_r) + 1$.
\item For every causally disconnected subset $\cs\subset\cc,\;\exists\, e_\alpha\in\cs$ such that ${\eta}(e_\alpha) = 1$.
\end{enumerate}
We will refer to this $k$-layer-map assignment as {\sl quasi-layers} (QL). Let $\cQnk$ denote the set of $k$-{\ql} orders (see Fig.~\ref{fig:qlorders}).
\end{definition}
{\begin{figure}[h]
\begin{center}
\begin{tabular}{c c c}
\includegraphics{DharOrder-a.pdf}
\includegraphics{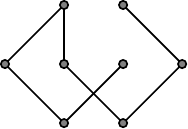}
\includegraphics{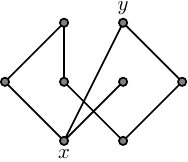}
\end{tabular}
\caption{The first two orders are  $3$-{\ql} orders while the third one is not a $k$-{\ql} order for {\it any}
  $k$ because element $x$ with $\eta(x)=1$ and element $y$ with $\eta(y)=3$ are related by a link thus violating
  Condition 1 of Def. \ref{def:ql}.\label{fig:qlorders}}
\end{center}
\end{figure}}

\begin{claim}\label{claim1}
There is a unique assignment of {\ql}s for any $\cc\in\cQnk$, i.e., every $k$-QL order is a unique labelled order in $\Omega_n$.% to the elements of $C \in \cQnk$ satisfying Properties 1 and 2 in Def.~\ref{def:ql}.
\end{claim}

\begin{proof}
Let $\cc\in\cQnk$ and let $\eta$ and $\eta'$ be two distinct {\ql} assignments on $\cc$. Thus $\exists\,e_r \in \cc$ such that ${\eta}(e_r) \neq  {\eta}'(e_r)$. Wlog let ${\eta}'(e_r) = {\eta}(e_r) + \eta_0$ with $\eta_0>0$. If $\cs \ni e_r$ is the unique irreducible causally disconnected subset containing $e_r$, it follows from Def.~\ref{def:ql}-1 that for every $e_s \in \cs$ which is linked to $e_r$, $\eta(e_s)=\eta(e_r)\pm 1$ and $\eta'(e_s)= \eta'(e_r)\pm 1, \Rightarrow$ ${\eta}'(e_s) = {\eta}(e_s) + \eta_0$ (where the choice ``$\pm$" depends on whether $e_s$ is to the future or past of $e_r$). Since $\cs$ is irreducible, every element in $\cs$ is ``connected" via a set of future and past relations to every other element in $\cs$. Therefore $\forall\; e_r\in \cs,\;{\eta}'(e_r) = {\eta}(e_r) + \eta_0$ and  since $\eta(e_r) \geq 1 \Rightarrow \eta'(e_r) \geq 1+\eta_0$. Since $\eta'$ is also a {\ql}, $\exists\,\,e_\alpha\in \cs$ such that $\eta'(e_\alpha)=1$, which is true only if $\eta_0=0$ thus implying that $\eta=\eta'$.
\end{proof}

$\cQnk$ is therefore a proper subset of $\Omega_n$. It is evident that $\cDnk \cap \cQnk \neq \emptyset$, but that one
is not nested inside the other. This is because $\exists\;\cc\in\cDnk$ in which there are links between non-consecutive
layers, and hence Def.~\ref{def:ql}-1 is not satisfied. Conversely $\exists\;\cc\in\cQnk$ such that the elements in the
$(i+2)$th {\ql}  are not all related to those in the
$i$th {\ql} and hence Def.~\ref{def:layer}-2 is not satisfied. Importantly, since the KR orders satisfy both Def.~\ref{def:ql}-1 and Def.~\ref{def:layer}-2, they lie in $\cDnk\cap \cQnk$.
We give examples of these differences in  Fig.~\ref{fig:dharorders}. Note also that typical manifold-like orders do not lie in either $\cQnk$ or $\cDnk$.
\begin{figure}[h]
\begin{center}
\begin{tabular}{c c c c}
\includegraphics{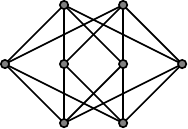}
\includegraphics{DharOrder-a.pdf}
\includegraphics{DharOrder-b.pdf}
\includegraphics{InQNotD.pdf}
\end{tabular}
\caption{The first two orders belong to both $\cDnk$ and $\cQnk$, the third to $\cDnk$ but not to $\cQnk$, and the fourth to $\cQnk$ but not to $\cDnk$.}\label{fig:dharorders}
\end{center}
\end{figure}

%\begin{definition}\label{def:bl}
%In a {\sl{\blue{bilayer}}} all relations are links.
%\end{definition}
{Another important example is that of {\sl{{bilayer}}} orders, in which all relations are links.} Clearly, every $\cc\in\cLn$ is a {bilayer} order. Conversely, to any {bilayer} order $\cc$, we can assign the {\ql} $\eta=1$ for all minimal elements and $\eta=2$ otherwise. This satisfies Def.~\ref{def:ql}-1 and \ref{def:ql}-2 and hence $\cc\in\cLn$.

Next, we define a class of layered orders that satisfy less stringent conditions than $k$-{\ql} orders:
\begin{definition}\label{def:pql}
In a {\sl{$k$-pseudo-quasi-layer ($k$-PQL)}} order $\cc\in\Omega_n$  it is possible to assign a $k$-layer-map
$\vartheta$, such that  Def.~\ref{def:ql}-1 is satisfied but not
necessarily Def.~\ref{def:ql}-2.  We will refer to this $k$-layer-map assignment as {\sl pseudo-quasi-layers} (PQL).
\end{definition}
%Let $\cPnk$ denote the set of $k$-{\pql} orders. It is evident that the same order can be obtained from different fillings of a given set of $k$-{\pql}s. Thus, while $\cPnk$ can be mapped into a subset of $\Omega_n$, it is not strictly a subset of $\Omega_n$.

We can also {\it build} up a  $k$-{\ql} or $k$-{\pql}  order in $\Omega_n$  from a  $k$-layer-map   on $n$ elements,
and then add in relations that satisfy the requirements of Def.~\ref{def:ql} or Def.~\ref{def:pql}, respectively. While
the former gives us a subset $\cQnk$ of orders in $\Omega_n$, the latter generates a set of orders
$\cPnk$ of  {\sl{$k$-{\pql} orders}} which  has redundancies and is hence strictly not a subset of $\Omega_n$. Namely,
the same order in $\Omega_n$ can be generated through different  assignments of  $k$-{\pql}s to the $n$ elements.

For {bilayer}s Def.~\ref{def:layer}-2 is trivially satisfied so that there is no distinction between $2$-layers and
{$2$-\pql}s. {This distinction however does become important in generalising the analysis to include higher layer orders.}
%Since we are interested in generalising the analysis to include $3$-level orders this distinction however does become important.

\subsection{Causal Set Actions}\label{sec:csaction}

The CST {\sts} over $\Omega_n$ is given by
\begin{equation}\label{eq:pf}
\cZ_n= \sum_{\cc \in \Omega_n} \exp\Big(\frac{i}{\hbar}{\mS}(\cc) \Big),
\end{equation}
where ${\mS}(\cc)$ denotes a choice of causal set action. The Kleitman-Rothschild result implies that if $\mS(\cc)=0$,
the KR orders dominate $\cZ_n$. The choice of $\mS(\cc)$ is therefore crucial in taming the contribution of the KR
orders. In analogy with the continuum, the natural choice for $\mS(\cc)$ is the discrete {Einstein-Hilbert} action,
or the $d$ dimensional {BDG} action \cite{benincasa,glaser1,glaser2}
%The {BDG} action in dimension $d$ for an order $C$ \cite{benincasa,glaser1,glaser2}  is given by
\begin{equation}\label{eq:bda}
\frac{1}{\hbar}\mathcal{S}_{BDG}^{(d)}(\epsilon,\cc) \equiv \mu(d,\epsilon) \Big(n + \sum_{j = 0}^{j_{max}(d)}
\lambda_{j}(d,\epsilon)N_{j}\Big),
\end{equation}
where $N_j$ is the number of $j$-element order intervals in $\cc$, $\epsilon$ is a new ``mesoscale" to suppress
fluctuations and $\mu(d,\epsilon)$, $\lambda_j(d,\epsilon)$ and $j_{max}(d)$ are dimension and $\epsilon$ dependent
constants  (see  \cite{benincasa,glaser1,glaser2} for details). In the limit $n\rightarrow\infty$ the expectation value of $\mS(\cc)$ over different Poisson sprinkling gives the {Einstein-Hilbert} action, upto boundary terms \cite{Buck:2015oaa,Dowker:2020xfg,Machet:2020axq}.

The ``link action" on the other hand depends only on the number of links $N_0$
\begin{equation}\label{linkact}
\frac{1}{\hbar}{\mathcal{S}}_{L}(\cc) \equiv \mu \Big(n + \lambda_{0}N_{0}\Big),
\end{equation}
which can also be obtained from the {BDG} action by putting $\lambda_j(d,\epsilon)=0,\,\,\forall\,\,j>0$.

The {\sts} Eqn.~\ref{eq:pf} can be split into a sum over mutually disjoint subsets $\{\pi_1,\pi_2,\ldots\}$ of $\Omega_n$ with $\sqcup_s \pi_s = \Omega_n$, so that
\begin{equation}
\cZ_n = \sum_{\pi_s}\cZ_n\Big|_{\pi_s},
\end{equation}
where $\cZ_n\Big|_{\pi_s}$ denotes the restriction of $\cZ_n$ to $\pi_s$. While such a split is obviously non-unique, it
allows us to isolate the contributions from specific classes of orders, like the KR orders. In \cite{carlip}, $\cZ_n$ was
restricted to the subset of {bilayer} orders.  In this work we will consider the {\sts} using the link action $\lpp$ and its restriction to the {\ql} and {\pql} orders.

\section{KR Orders and the Link Action}\label{sec:3p}
%\section{Suppressing KR orders with $\mathcal{S}_l$}\label{sec:3p}
%\section{Contributions to \red{${\cal{Z}}_l$} from KR-like orders}\label{sec:3p}
In this section we compute the leading order contribution to the {\sts} restricted to three different subsets of $\Omega_n$ each of which contains what we will loosely refer to as ``KR-like" orders.

First we will need a few more definitions. The set $\cPnk$ is not in one to one correspondence with elements in $\Omega_n$ since the same order in $\Omega_n$ can be obtained from different fillings of a given set of $k$-{\pql}s. Therefore we cannot use $\cPnk$ to partition $\lpp$. Instead we look at a subset $\cP^*_{\vec{q},n}\subset\cPnk$ defined as follows.

Let $\Ppql^k$ denote the set of all of possible assignments of {\pql}s over the set of $n$ integers. Since there are $k^n$ ways of making a {\pql} assignment, $|\Ppql^k|=k^n$. Any element of $\Ppql^k$ can be labelled by the {\sl{filling fraction}}, $\vec{q} = \left(q_1, q_2, \dots, q_k \right)$ where ${q_i}n$ is the cardinality of the $i^{\mathrm{th}}$ {\pql}, so that $\sum_{i=1}^k q_i = 1$. Since the elements are distinguishable, each choice of $\vec{q}$ comes with a multiplicity $m\left(\vec{q}\right)$ given by
\begin{equation}\label{eq:mq}
m(\vq) = \frac{n!}{(q_1n)!(q_2n)!\ldots (q_kn)!} \Rightarrow \sum_{\vec{q}} m\left(\vec{q}\right) = k^n.
\end{equation}
\begin{definition}\label{def:nla}
A {\pql} assignment $\vqs\in\Ppql^k$ is said to be {\sl{naturally labelled}} if the first {\pql} is $\{e_1,\ldots,e_{q_1 n}\}$, the second {\pql} is $\{ e_{{q_1}n + 1},\ldots, e_{({q_1 + q_2})n}\}$ and so on, and therefore is unique. We denote the set of naturally labelled $k$-{\pql} orders on $\vqs$ by $\cP^*_{\vq,n}\subset \cPnk$.
%$\cP^*_{\vq,n}\subset \cPnk$ is the collection of all $k$-{\pql} orders on $\vqs$.
\end{definition}
%\begin{definition}\label{def:nla}
%$\Ppql^{\vec{q}*}\subset\Ppql^{\vec{q}}$ is the {\sl{naturally labelled assignment}}, where for all $e_r,e_s\in \Ppql^{\vec{q}*}, \vartheta(e_r)<\vartheta(e_s)\Rightarrow r<s$. In other words, for each $\vq$, we consider the unique filling such that the first {\pql} contains the elements $\{e_1, e_2, \ldots, e_{{q_1}n}\}$, the second {\pql} the elements $\{ e_{({q_1}n + 1)}, e_{({q_1}n + 2)}, \ldots, e_{({q_1 + q_2})n}\}$ and so on.
%\end{definition}

We further denote the set of all $k$-{\pql} orders on $\vqs$ with $pn^2$ links by $\cP^*_{\vec{q},p,n}\subset \cP^*_{\vq,n}$ and similarly the set of all $k$-{\ql} orders with $pn^2$ links by $\cQpnk\subset\cQnk$.
\begin{claim}\label{claim2}
There exists a one-to-one map from $\cP^*_{\vec{q},p,n}$ to $\cQpnk$, so that $\cP^*_{\vq,p,n}\subset\cQpnk\subset\Omega_n$ and therefore $\left|\cP^*_{\vec{q},p,n}\right| \leq \left|\cQpnk\right|$.
\end{claim}
\begin{proof}
Let $\cs = \{\cs_1, \cs_2,\ldots, \cs_\kappa\}$ denote the set of all irreducible causally disconnected subsets of $\cc \in \cP^*_{\vec{q},p,n}$. Let $t_i$ denote the smallest {\pql} in $\cs_i$. Make the new {\pql} assignment
\begin{equation}
{\eta}(e_r) \equiv  {\vartheta}(e_r) - t_i + 1, \quad \quad \forall e_r \in \cs_i.
\end{equation}
Under $\eta$, consecutive layers are mapped to consecutive layers. Hence Def.~\ref{def:ql}-1 is still satisfied. Moreover, for each $\cs_i$, there exists an $e_r$ such that $\vartheta(e_r) = t_i$ and hence, $\eta(e_r) = 1$. This ``shuffling down" brings all the elements in the {\pql} $t_i$ associated with $\cs_i$ to the {\pql} $\eta=1$, so that the number of {\pql}s either stays the same or decreases. Thus, $\eta$ satisfies Def.~\ref{def:ql}-2 and hence is also a {\ql}. $\eta$ therefore defines a one-to-one map $\varphi: \cP^*_{\vq,p,n} \rightarrow \cQpnk$. Since for any pair of distinct $\cc, \cc' \in \cP^*_{q, p, n}$, wlog, $\exists\,e_r,e_s\in \cc,\cc'$ such that $e_r\link e_s$ in $\cc$ but not in $\cc'$. Since Def.~\ref{def:ql}-1 is still satisfied under $\eta$, $e_r\link e_s$ in $\varphi(\cc)$ but not in $\varphi(\cc')$. Therefore $\cc\neq \cc'\Rightarrow \varphi(\cc)\neq\varphi(\cc')$, which ensures that the map $\varphi$ is one-to-one.
\end{proof}

%Since $\cP^*_{\vec{q},p,n}$ and $\cQpnk$ for different $p$ are non-overlapping. We have a one-to-one map $\varphi:\cP^*_{\vec{q},n}\rightarrow \cQnk$. Hence $\cP^*_{\vec{q},n}\subset \cPnk$ is a class of orders in $\Omega_n$.

\subsection{Review of the Results on the Bilayer Orders }\label{sec:revlc}
We start with a brief review of the results of \cite{carlip} on {bilayer} orders rephrased in the language of
{\ql}s and {\pql}s.  For {bilayer} orders, the maximum number of possible links is $n^2/4$ and occurs only for the
filling fraction $\vq_b=(1/2,1/2)$. A clever bounding argument used by the authors (and which we will also employ in Sec.~\ref{section:klevels}) shows that to leading order
\begin{equation}\label{eq:lnqpn2}
\ln|\cLpn| = \ln|\cP^*_{\vq_{b},p,n}| + o(n^2),\quad |\cP^*_{\vq_{b},p,n}| = \binom{n^2/4}{pn^2},
\end{equation}
where $N_0=pn^2 \equiv \tilde{p}n^2/4$ and $p\leq 1/4$. Using Stirling's approximation to leading order in $n$
\begin{equation}\label{eq:plogbl}
\ln|\cP^*_{\vq_{b},p,n}| = \frac{n^2}{4} h(\tilde{p}) + o(n^2),
\end{equation}
where
\begin{equation}\label{entrfunc}
h(\tilde{p}) = -\tilde{p}\ln{\tilde{p}} - (1 - \tilde{p})\ln{(1 - \tilde{p})}, \quad \quad 0 < \tilde{p} < 1,
\end{equation}
is Dhar's entropy function \cite{dhar1}.

For $p \neq p'$, $\cLpn$ and $\mathcal{Q}^2_{p',n}$ are disjoint, and one can therefore partition $\cLn$ into disjoint subsets labelled by $p$. The {\sts} for the {BDG} action or equivalently the link action is then
\begin{equation}\label{finalpartfunc}
\mathcal{Z}_n^{(BDG)}\Big|_{\cLn}=\lpp\Big|_{\cLn}=\int_0^1 \mathrm{d}\tilde{p} \exp{\left[\frac{n^2}{4}\left(i\mu\lambda_0\tilde{p} + h(\tilde{p})\right)+o(n^2)\right]}.
\end{equation}
The integral in Eqn.~\eqref{finalpartfunc} was evaluated to leading order using the method of steepest descent. For the
{BDG} action, $\mu\lambda_0<0$ it was shown  that $\mathcal{Z}_n^{(BDG)}\Big|_{\cLn}$ is exponentially suppressed when
\begin{equation}\label{eq:parange}
\tan\left(\frac{\mu\lambda_0}{2}\right) < -\sqrt{\frac{27}{4}e^{-1/2}-1}.
\end{equation}

For $\lpp\Big|_{\cLn}$ on the other hand,  it is possible for $\mu\lambda_0>0$ and hence it  is exponentially suppressed when
\begin{equation}\label{eq:parange2}
\tan\left(\frac{\mu\lambda_0}{2}\right)>\sqrt{\frac{27}{4}e^{-1/2}-1}.
\end{equation}
%It is important to remark here that this solution is valid only for the integral from $\tilde{p}=0$ to $1$. Since small $\tilde{p}$ values give entropically small contributions, as for example the $n-1$ antichain with a simple relation for the remaining $n^{th}$ element; one might ask how this result change for the \red{integration} range $[\tilde{p}_0,1]$ if $\tilde{p}_0>0$.
%\red{[[Discussion on how dimension comes into play]]}
%In {BDG} action for a given spacetime dimension $d$, $\lambda_js$ are known constants whereas $\mu = \xi_d l^{d-2}$, where $\xi_d$ is a known constant and ``$l$" is a free parameter with the dimension of length. For $d=2$, $\mu\lambda_0=-4$ \cite{Sorkin:2007qi}, which implies that $\tan\left(\mu\lambda_0/2\right) = 2.185\;>-\sqrt{3}$, and therefore $\cLn$ is not suppressed by the $2d$ {BDG} action. However, we see ``link action" as an interesting action in itself and not merely as an approximation to the {BDG} action. The spacetime dimension, therefore has no role to play in determining the parameters of the ``link action", and hence we can always treat them as free parameters.}

It is interesting to note the dimension dependence of the suppression in the {BDG} action, where $\mu\lambda_0=-\left(\frac{l}{l_p}\right)^{d-2}\beta_d C_1^{(d)}$ where $l$ is the discreteness scale \cite{glaser1,glaser2}. In $d=2$, for example, $\mu\lambda_0 = -4$ and hence there is no suppression. For all $d\geq 3$, however, by adjusting $l/l_p$, one can find a suitable suppression regime. In $d=4$, in particular, there is a suppression for all $l\gtrapprox 1.452\,l_p$.

We note that the restricted {\sts} in the  calculation of \cite{carlip} also includes {bilayer} orders which are {\it{not}}
entropically dominant, for example the $(n-1)$-element antichain with just one relation with the $n^{\mathrm{th}}$
element. These nevertheless seem important in estimating the leading order contribution to the CST {\sts}. This suggests
that the restriction to appropriate larger subsets of $\Omega_n$ may be crucial in obtaining the right semi-classical
approximation of the CST {\sts}.

\subsection{Naturally Labelled $k$-{\pql} Orders with a Fixed Filling Fraction}\label{sec:3pql}

We now consider the restriction of $\lpp$ to ${\cal{P}}^*_{\vq, n}\subset\Omega_n$, which is the set of naturally
labelled {\pql} orders with fixed filling fraction $\vq$. For the choice $\vq_{kr}=(1/4,1/2,1/4)$,
${\cal{P}}^*_{\vq_{kr}, n}$ includes the set of naturally labelled ``typical" KR orders. By this we mean
Def.~\ref{def:kr}-1 and Def.~\ref{def:kr}-3 are satisfied exactly, i.e., without the $o(n)$
fluctuations. Fig.~\ref{fig:palpha} is an example of such an order for $n=8$.

\begin{figure}[!bhtp]
\begin{center}
\includegraphics{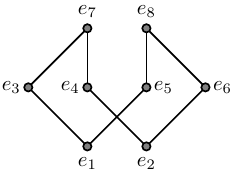}
\caption{A naturally labelled $n=8$  ``typical''  KR order with $\vq=\vq_{kr}$.}
\label{fig:palpha}
\end{center}
\end{figure}

The maximum number of links possible for any  $\cc\in\cP^*_{\vq,n}$ is
\begin{equation}\label{eq:nmax}
N_{max} = (q_1q_2+q_2q_3+\dots+q_{k-1}q_k)n^2=:\alpha(\vq) n^2,
\end{equation}
and therefore we can express the number of links $N_0$ for $\cc\in\cP^*_{\vq,n}$ by
\begin{equation}
N_0 = \tilde{p}N_{max} = pn^2, \qquad 0 \leq \tilde{p} \leq 1.
\end{equation}
We refer to $\tilde{p}$ as the {\sl{linking fraction}}. As shown in Fig.~\ref{fig:n8p4}, for a given $p$
%As in {\lc}, we calculate consider first the contribution to $\cP^*_{\vq,n}$ for fixed $\tilde{p}$. $|\cP^*_{\vq,p,n}|$ can be calculated from $|\cP^*_{\vq,\alpha(\vq),n}|$ as follows. In Fig.~\ref{fig:n8p4}, all the possible links are coloured black while those that can be realised for a fixed $\tilde{p}$ are coloured orange. There are $\binom{\Nm}{N_0}$ ways of doing so, which means that
\begin{equation}\label{eq:cardp}
\left|{\cal{P}}^*_{\vec{q}, p, n}\right|= \binom{\Nm}{N_0}= \binom{\alpha(\vq)n^2}{pn^2},
\end{equation}
which to leading order in $n$ is
\begin{equation}\label{logC}
\ln\left|\cP^*_{\vec{q},p, n}\right| = \alpha(\vq) n^2 h(\tilde{p}) + o(n^2).
\end{equation}
This is identical to  Eqn.~\ref{eq:plogbl} with $n$ replaced by $2\sqrt{\alpha(\vq)}n$, so that
\begin{equation}\label{finalpartfunc2}
\lpp\Big|_{\cP^*_{\vq,n}} = \int_0^1 \mathrm{d}\tilde{p} \exp{\left[\alpha(\vq) n^2\left(i\mu\lambda_0\tilde{p} + h(\tilde{p})\right)+o(n^2)\right]},
\end{equation}
which is the same as $\lpp\Big|_{\cLn}$ upto leading order. Thus, from the analysis in \cite{carlip}, which is unaffected by any rescaling of $n$, we see that $\lpp\Big|_{\cP^*_{\vq,n}}$ is exponentially suppressed for $\mu\lambda_0$ given by Eqn.~\eqref{eq:parange} and Eqn.~\eqref{eq:parange2}. The result is somewhat surprising, since it means that the {\sts} in \cite{carlip} captures a more general (leading order) feature of the full CST {\sts}.
\begin{figure}[!bhtp]
\begin{center}
\includegraphics{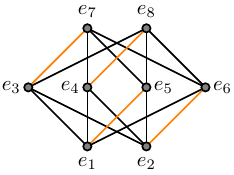}
\caption{A naturally labelled $3$-{\pql} order with $n = 8, \vq=\vq_{kr}, \tilde{p} = 0.25$. All the possible links are coloured black while those that are realised for this $\tilde{p}$ are coloured orange.}
\label{fig:n8p4}
\end{center}
\end{figure}

What we have calculated is the leading order contribution of $\cP^*_{\vq,n}$ to $\lpp$ for any $\vq$. Since the $\cP^*_{\vq,n}$ for different choices of $\vq$ overlap, one cannot however, further sum over $\vq$. Of special importance is the set $\cP^*_{\vq_{kr},n}$, which includes the naturally labelled ''typical" KR-orders. Since the contribution from KR orders near typicality could also be entropically important, we now look for a larger class of orders which contain all of the KR orders.

\subsection{$k$-{\ql} Orders}\label{section:klevels}
We now consider the more general subset $\cQnk$, which includes all $k'$-{\ql} orders for $k'\leq k$. In particular for $k\geq 3$, this includes all KR orders.

Again we begin with the subset $\cQpnk$. For $p \neq p'$, $\cQpnk$ and ${\cal{Q}}^k_{p',n}$ are disjoint, and one can therefore partition $\cQnk$ into the disjoint subsets labelled by $p$ and use $p$ as an integration factor as before. % Thus, in the partition function restricted to the class $\cQnk$, one could sum over $p$.
%\begin{equation}
%\lpp\Big|_{\cQnk} \equiv \int_0^1 \mathrm{d}p\, |\cQpnk| \exp{(iS_l(pn^2))}.
%\end{equation}
Unlike our earlier calculation, however, $|\cQpnk|$ is harder to obtain directly. $|\cLpn|$ was obtained in
\cite{carlip} by saturating both a lower and an upper bound. We employ these same methods here.

Let $\cP^k_{p,n}\subset\cPnk$ be the set of all $k$-{\pql} orders with $pn^2$ links. Since $\cP^k_{p,n}$ contains different {\pql} labellings of the same order, $\cQpnk\subset\cP^k_{p,n}$. From Claim \ref{claim2}, for any filling fraction $\vq$,
\begin{equation}\label{eq:ineq1}
|\cP^*_{\vq,p,n}|\leq|\cQpnk|\leq|\cP^k_{p,n}|.
\end{equation}
In order to tighten these bounds, we vary over $\vq$ and see that $\alpha(\vq)$ takes the maximum value $\alpha_{m} = 1/4$ for $\vq_x=(1/4-x,1/2,1/4+x)$ where $-1/4\leq x\leq 1/4$. This includes the two configurations: the symmetric $k=2$ case, with $\vq_b=(1/2,1/2)$ and the ``typical KR" $k=3$ case, with $\vq_{kr}=(1/4,1/2,1/4)$. Since $|\cP^*_{\vec{q},p,n}|$ is a monotonically increasing function of $\alpha(\vq)$ for fixed $n$ and $p$, it achieves a maximum at $\alpha_m$. Let $\vec{q}_0$ denote one of these maximising configurations.

Thus
\begin{equation}\label{eq:ineq2}
|\cP^k_{p,n}|=\sum_{\vq} m(\vq) |\cP^*_{\vec{q}, p, n}| \leq \sum_{\vec{q}} m(\vq) \left|\cP^*_{\vq_0, p, n}\right| = k^n \left|\cP^*_{\vq_0, p, n}\right|,
\end{equation}
where $m(\vq)$ is defined in Eqn.~\ref{eq:mq} which implies that
\begin{equation}\label{cardineq3}
|\cP^*_{\vec{q_0}, p, n}| \leq |\cQpnk| \leq k^n |\cP^*_{\vec{q_0}, p, n}|.
\end{equation}
Note that this reduces to the calculation of \cite{carlip} when $k=2$. The factor $k^n$ contributes only to the subleading order of $n$ and therefore
\begin{equation}\label{eq:lnqpnk}
\ln|\cQpnk| = \ln|\cP^*_{\vec{q_0}, p, n}| + o(n^2).
\end{equation}
Thus we see that the leading order contribution to $\lpp\Big|_{\cQnk}$ for {\it{any}} $k\geq 3$ comes from a subset of $\Omega_n$ which includes the KR orders as well as the symmetric {bilayer} orders. Moreover, it is again the same as $\lpp\Big|_{\cLn}$ to leading order. Thus we see that the link action serves to suppress the contribution from all KR orders for $\mu\lambda_0$ satisfying Eqn.~\eqref{eq:parange} and \eqref{eq:parange2}.

\subsection{``KR + dust"}\label{sec:krlike}
In order to illustrate the non-triviality of the previous calculations, we consider a wholly different class of orders, which contains ``typical KR" orders of all cardinalities $n'<n$, along with the ``dust" of an $(n-n')$-element antichain.
\begin{definition}\label{def:cchin}
Let $\cc\in\mathcal{Q}^3_n$ admit a partition $\cc=\cs_1\sqcup \cs_2$ so that $\cs_1$ is a $\chi n$ order and $\cs_2$ is a $(1-\chi)n$-element antichain, where $0\leq\chi\leq 1$ so that:
\begin{enumerate}%[label=Property \arabic*.,itemindent=*]
\item $e_r\in \cs_2 \Rightarrow \eta(e_r)=1$.
\item The $\chi n$ elements of $\cs_1$ are assigned {\ql}s with filling fraction $\vq=(1/4,1/2,1/4)$.
\item Each element in $\cs_1$ with $\eta=1$ is linked to exactly ${\chi}n/4$ elements with $\eta=2$, and similarly each element with $\eta=3$ is linked to exactly ${\chi}n/4$ elements with $\eta=2$.
\item $e_r\in \cs_1$ is minimal $\Leftrightarrow \eta(e_r)=1$.
\end{enumerate}
Let $\krd_{\chi,n}$ denote the set of these ``KR + dust" orders, for a given $\chi$. Fig.~\ref{fig:krlike24} shows an example of such an order.
\end{definition}
Note that unlike the previous two cases, this fixes not only the filling fraction $\vq$, but also the number of links $N_0=(\chi n)^2/8$ in $\krd_{\chi,n}$. For $\chi \neq \chi'$, $\krd_{\chi,n}$ and $\krd_{\chi',n}$ are disjoint and therefore $\krd_n = \sqcup_{\chi} \krd_{\chi,n} \subset \mathcal{Q}^3_n$.
%\begin{definition}
%$\krd_n\subset \mathcal{Q}^3_n$ is the disjoint union of $\krd_{\chi,n}$ for all $\chi\in[0,1]$.
%\end{definition}
Thus
\begin{equation}
\lpp\Big|_{\krd_n} \equiv \int_0^1 \mathrm{d}\chi\, |\krd_{\chi,n}| \exp{(i\mS_L(\chi,n))}.
\end{equation}
\begin{figure}[h]
\begin{center}
\includegraphics{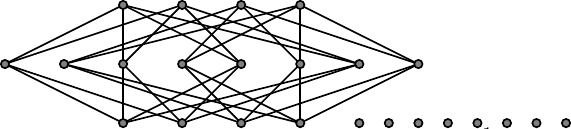}
\caption{A 24-element ``KR + dust" order with $\chi = 2/3$.}
\label{fig:krlike24}
\end{center}
\end{figure}

In order to evaluate $|\krd_{\chi,n}|$ we consider the larger set of orders $\krd'_{\chi,n}\supset \krd_{\chi,n}$ which satisfy Def.~\ref{def:cchin}-1-3 but not necessarily Def.~\ref{def:cchin}-4. Therefore $\krd'_{\chi,n}\subset\cP^3_n$ since the elements of any $\cc\in\krd'_{\chi,n}$ can be assigned {\pql}s but not necessarily {\ql}s. Let $\mathcal{G}_{\chi,n}\subset\krd'_{\chi,n}$ such that for every $\cc\in \mathcal{G}_{\chi,n}$, there is at least one element in the second {\pql} which is not linked to any element in the first {\pql}. From Def.~\ref{def:cchin} we see that $\krd_{\chi,n} = \krd'_{\chi,n}\backslash \mathcal{G}_{\chi,n}$ and therefore

\begin{equation}\label{eq:cardccn}
|\krd_{\chi,n}| = |\krd'_{\chi,n}| - |\mathcal{G}_{\chi,n}|.
\end{equation}
Let us start with computing $|\krd'_{\chi,n}|$. The number of ways of separating out $(1-\chi)n$ elements for $\cs_2$ is $\binom{n}{\chi n}$ and the number of ways of distributing the remaining ${\chi}n$ elements so that it satisfies Def.~\ref{def:cchin}-2 is $\binom{\chi n}{\chi n/4}\times \binom{3\chi n/4}{\chi n/4}$. Additionally, the number of ways of linking the elements so that Def.~\ref{def:cchin}-3 is satisfied is $\binom{\chi n/2}{ \chi n/4}^{{\chi}n/2}$. Hence,
\begin{equation}
\left|\krd'_{\chi, n}\right| = \binom{n}{\chi n} \binom{\chi n}{\chi n/4} \binom{3{\chi}n/4}{{\chi}n/4} \binom{{\chi}n/2}{{\chi}n/4}^{{\chi}n/2}.
\end{equation}
Now we compute $|\mathcal{G}_{\chi,n}|$. The number of ways of separating out $(1-\chi)n$ elements for $\cs_2$ and the number of ways of distributing $\chi n$ elements so that it satisfies Def.~\ref{def:cchin}-2 is the same as that for $|\krd'_{\chi,n}|$ and therefore
\begin{equation}
|\mathcal{G}_{\chi,n}| = \binom{n}{\chi n} \binom{{\chi}n}{{\chi}n/4} \binom{3{\chi}n/4}{{\chi}n/4} A_{\chi,n},
\end{equation}
where
\begin{equation}
A_{\chi,n} \equiv \frac{{\chi}n}{2} \binom{{\chi}n/2-1}{{\chi}n/4}^{{\chi}n/4} \binom{{\chi}n/2}{{\chi}n/4}^{{\chi}n/4} = \frac{{\chi}n}{2} \cdot 2^{-{\chi}n/4} \binom{{\chi}n/2}{{\chi}n/4}^{{\chi}n/2},
\end{equation}
denotes the number of ways in which the links can be assigned, in order to satisfy Def.~\ref{def:cchin}-3, with at least one element in the second {\pql} not linked to any element in the first {\pql}. Thus
\begin{equation}
|\krd_{\chi,n}|= \binom{n}{\chi n} \binom{\chi n}{\chi n/4} \binom{3{\chi}n/4}{{\chi}n/4} \binom{{\chi}n/2}{{\chi}n/4}^{{\chi}n/2}\left(1- \frac{{\chi}n}{2} \cdot 2^{-{\chi}n/4} \right).
\end{equation}
In the limit of large $n$, the second term is highly suppressed so that
\begin{eqnarray}
%\begin{split}
\left|\krd_{\chi, n}\right|& \approx& \binom{n}{\chi n} \binom{\chi n}{\chi n/4} \binom{3{\chi}n/4}{{\chi}n/4} \binom{{\chi}n/2}{{\chi}n/4}^{{\chi}n/2}\nonumber\\
&\approx& \frac{n!}{((1-\chi)n)!} \biggl( ({\chi}n/2)!\biggr) ^{{\chi}n/2 - 1} \biggl( ({\chi}n/4)!\biggr)^{-({\chi}n + 2)}.
%\end{split}
\end{eqnarray}
Using Stirling's approximation,
\begin{equation}
\ln{\left|\krd_{\chi,n}\right|} = \frac{({\chi}n)^2}{4} \ln{2} - \frac{{\chi}n}{4} \ln{({\chi}n)} + o(n\ln n),
\end{equation}
and hence
\begin{eqnarray}\label{eq0.3}
%\begin{split}
\lpp\Big|_{\krd_n} &=& \int_0^1 \mathrm{d}{\chi} \left|\krd_{\chi,n}\right| \exp\Big(i\mu \Big(n + \lambda_0\frac{({\chi}n)^{2}}{8} \Big)\Big)\nonumber\\
%\approx e^{i{\mu}n} \int_0^1 \mathrm{d}{\chi} \exp\Big(\frac{({\chi}n)^2}{4} \Big(\ln{2} - \frac{i\mu\lambda}{2}\Big) - \frac{{\chi}n}{4} \ln{({\chi}n)}+o(n\ln n)\Big)
&=& \int_0^1 \mathrm{d}{\chi} \exp\Big(\frac{({\chi}n)^2}{4} \Big(\ln{2} + \frac{i\mu\lambda_0}{2}\Big) +o(n^2)\Big)\nonumber\\
&\approx&-i\frac{\sqrt{\pi}}{n}\;\frac{\mathrm{erf}\left(i\frac{n}{2}\sqrt{\ln2+i\frac{\mu\lambda_0}{2}}\right)}{\sqrt{\ln2+i\frac{\mu\lambda_0}{2}}}.
%\end{split}
\end{eqnarray}
This is divergent in the limit of large $n$ for {\it{any}} range of the parameters $\mu$ and $\lambda_0$ which means that $\krd_n\subset\Omega_n$ is not suppressed in the CST {\sts}, even though it contains ``typical" KR orders with $n'\leq n$. This underlines the importance of the choice of subset to which $\lpp$ is restricted.

\section{Conclusions}\label{sec:con}
In this work we have shown that the link action can suppress the entropy of KR orders in the CST {\sts} $\lpp$ using
techniques very similar to those used in \cite{carlip} to show the suppression of {bilayer} orders. In particular, the leading
order contributions to $\lpp$ in Sec.~\ref{sec:3pql} and Sec.~\ref{section:klevels} are shown to be the same as that in
\cite{carlip}. In the calculation of Sec.~\ref{section:klevels}, this can be traced to the fact that the $N_{max}$ is
maximised when the filling fraction is $\vq=(1/4-x,1/2,1/4+x),\;-1/4\leq x\leq 1/4$, which includes the symmetric
{bilayer} orders. Thus it seems that the calculations  for {bilayer} orders in \cite{carlip} already captures the essence of the contribution
from the KR orders!

We have also examined the contribution to $\lpp$ of naturally labelled $k$-{\pql} orders for any $k$ with fixed $\vq$. As shown in Sec.~\ref{sec:3pql}, to leading order this too reduces to $\lpp\Big|_{\cLn}$ and is suppressed for the same range of parameters.

In order to emphasise the non-triviality of these results, we also considered  the restriction of $\lpp$ to the ``KR +
dust" subset of $\Omega_n$. We find that there is {\it{no}} parameter range in which this contribution is suppressed. This example illustrates the importance of the choice of subset to which the CST {\sts} is restricted and suggests that even to leading order in $n$, subtle cancellations of the phases are important.

It would be of interest to extend these results to the other actions, like the relational action, or better still the
{BDG} action. In order to perform a similar calculation for the {BDG} action, one would have to count the
class of iso-action $k$-{\ql} orders rather than those with fixed $N_j$.  Even a leading order estimation would be of
great value, but this is beyond the scope of the present work.

We now present some numerical evidence that supports the idea suggested in \cite{carlip} that our link action result may be
relevant to leading order for KR orders even for the {BDG} action. {We generate an ensemble of naturally
  labelled $3$-{\pql} orders with different filling fractions $\vq$ including $\vq_{kr} = (1/4,1/2,1/4)$. The orders for
  each $\vq$ are generated by randomly choosing the links for a given linking fraction $\tilde{p}$. We examine the
  expectation values of the ratios of the $j$-element order interval $N_j$ for $j=1,2$ and $3$, to the links $N_0$ as a
  function of the linking fraction $\tilde{p}$. Typical KR orders have a linking fraction $\tilde{p}\sim 1/2$ and a
  filling fraction of $\vq_{kr}$. From Fig.~\ref{fig:6}, we
  see that  as  $ \tilde{p}$ approaches $ \frac{1}{2}$ from below,  $\langle N_j/N_0\rangle\rightarrow 0$ for $j=1,2,3$,
  for all these choices of $q$, with the fastest decay for $\vq_{kr}$. Moreover, this behaviour improves substantially
  with  $n$, as seen by comparing Fig.~\ref{fig:6}  and  \ref{fig:7}.   This suggests that the BDG
  action for KR orders can be replaced by the link action.}

\begin{figure}[h!]
\begin{center}
\includegraphics[height=10.5cm]{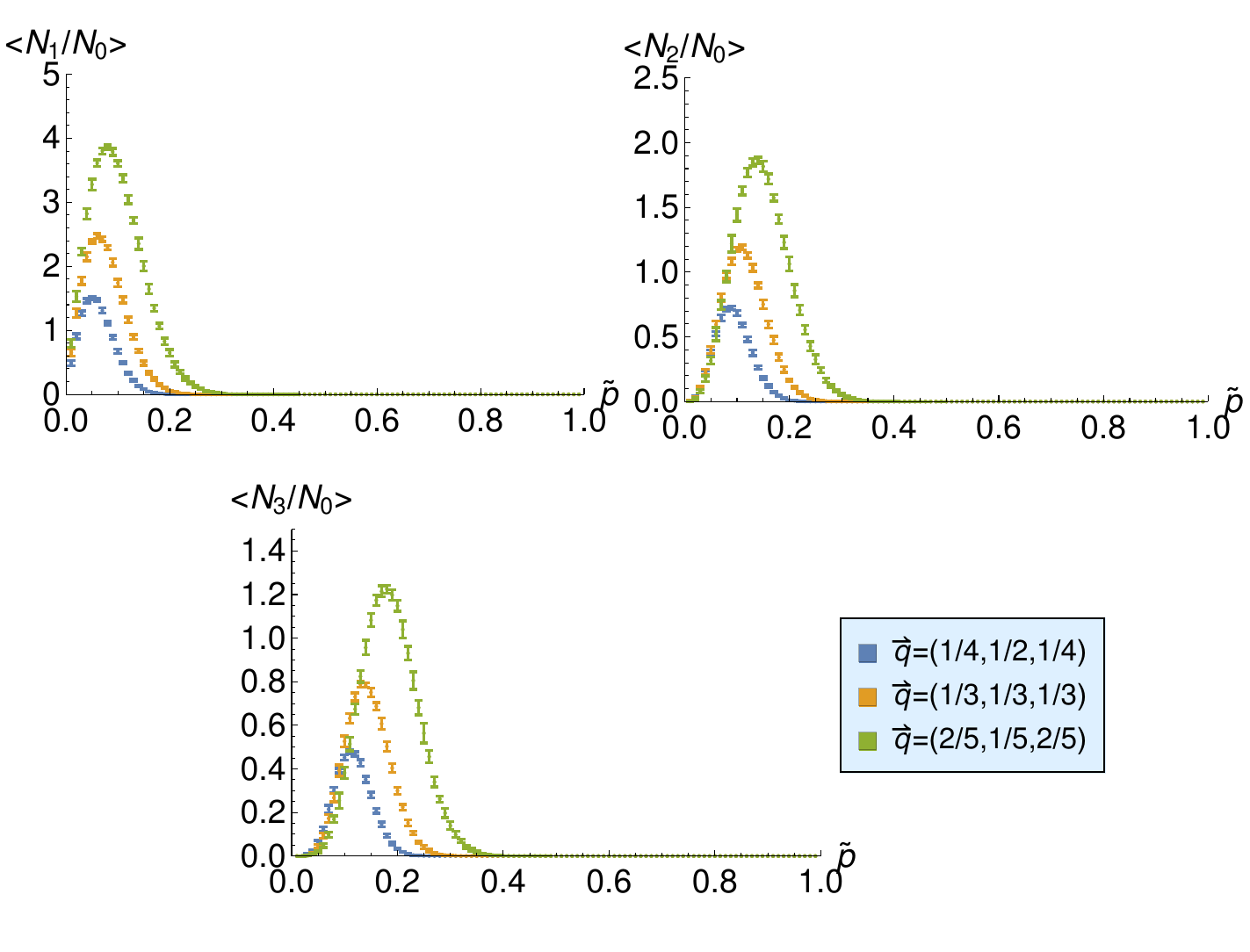}
\caption{The ratios $\langle N_1/N_0\rangle$, $\langle N_2/N_0\rangle$ and $\langle N_3/N_0\rangle$ vs the linking fraction $\tilde{p}$ for $n\approx 400$ for three different types of 3-{\pql} orders $\vq_{kr}$, $\vq=(1/3,1/3,1/3)$ and $\vq=(2/5,1/5,2/5)$. The average is taken over twenty samples.}\label{fig:6}
\end{center}
\end{figure}
\begin{figure}[h!]
\begin{center}
\includegraphics[height=5cm]{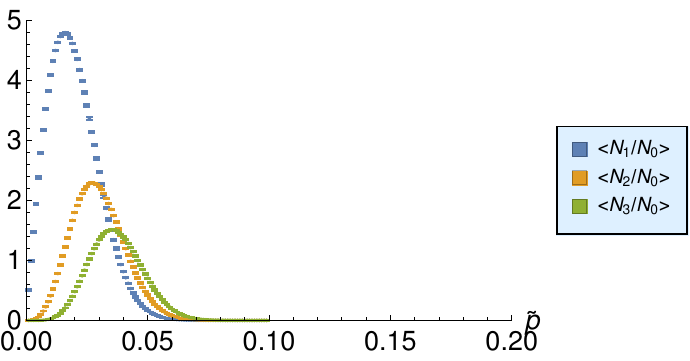}
\caption{The ratios $\langle N_1/N_0\rangle$, $\langle N_2/N_0\rangle$ and $\langle N_3/N_0\rangle$ vs the linking fraction $\tilde{p}$ for n=4000 with $\vq_{kr}$. The average is taken over twenty samples.}\label{fig:7}
\end{center}
\end{figure}

%In Fig.~\ref{fig:6} and \ref{fig:7}, we show the behaviour of the ratio of the number of order intervals $N_j$ for $j>0$ to the number of links $N_0$, with the linking fraction $\tilde{p}$. We consider three different filling fractions $\vq$ with $n\approx 400$ in Fig.~\ref{fig:6}. The orders for each $\vq$ are generated by randomly choosing the links for a given linking fraction $\tilde{p}$. The choice $\vq_{kr}$ (in blue) includes the KR orders for which $\tilde{p}\sim 1/2$. The ratios $\langle N_j/N_0\rangle$ for $j=1,2,3$ can be seen to go to zero rapidly with $\tilde{p}$. In Fig.~\ref{fig:7} we show the same effect enhanced for a much larger $n$ of $4000$ with $\vq_{kr}$.

Noting that the {BDG} action is linear in $N_i$,   it is also of interest to see how non-linearities could affect our analysis. For a non-linear link action, any term proportional to $N_0^m=p^mn^{2m}$ with $m>1$ will dominate the $n^2$ order terms in the large $n$ limit, including the entropy, as is evident from Eqn.~\ref{finalpartfunc}. This gives a suppression for any choice of coefficient in the large $n$ limit, albeit not exponentially. This suggests that non-linear corrections to the {BDG} action could play an important role in entropy suppression in the deep quantum regime. This non-linearity must however give rise to the linear {BDG} action in the low energy limit of the theory.

In CST, the full sample space $\Omega_n$ also contains causal sets which are approximated by continuum spacetimes of dimension $d\neq 4$. If, for example, Kaluza-Klein spacetimes of dimension $4+D$ are the right continuum approximation of CST in the low energy limit, one has to explain why causal sets which approximate spacetimes of other dimensions are suppressed. It is possible that numerical investigations might provide useful clues, and are a concrete way forward.
\\\\
{\bf\large Acknowledgements}

SS is supported in part by a Visiting Fellowship at the Perimeter Institute.  Re-search at Perimeter Institute is supported in part by the Government of Canada through the  Department  of  Innovation,  Science  and  Economic  Development  Canada  and  by  the Province of Ontario through the Ministry of Colleges and Universities. AAS was supported in part by the Visiting Scholar Program at the Raman Research Institute.

\bibliography{mss}
\bibliographystyle{ieeetr}

\end{document}